\documentclass[11pt]{article}
\usepackage[margin=1in]{geometry}
\usepackage[T1]{fontenc}
\usepackage{lmodern,amsmath,amssymb,amsthm,mathtools,microtype}
\usepackage{booktabs,tabularx,enumitem}
\usepackage[hidelinks]{hyperref}
\hypersetup{pdftitle={Hardness Conjectures as Information Constraints: Does Complexity Theory Need New Axioms?},pdfauthor={Hunter Monroe}}
\newcommand{\Con}{\operatorname{Con}}
\newcommand{\Prf}{\operatorname{Prf}}
\newcommand{\Sim}{\operatorname{Sim}}
\newcommand{\EA}{\mathsf{EA}}
\newcommand{\PH}{\mathsf{PH}}
\newcommand{\NP}{\mathsf{NP}}
\newcommand{\Pclass}{\mathsf{P}}
\newcommand{\KT}{\mathsf{KT}}
\newcommand{\BB}{\mathsf{BB}}
\newcommand{\N}{\mathbb N}
\newcommand{\Sbase}{S^1_2}
\newcommand{\ExpAx}{\mathsf{Exp}}
\newcommand{\csim}{\preccurlyeq_{\mathrm{con}}}
\newcommand{\R}{R^{\log}}
\newtheorem{theorem}{Theorem}[section]
\newtheorem{lemma}[theorem]{Lemma}
\newtheorem{proposition}[theorem]{Proposition}
\newtheorem{corollary}[theorem]{Corollary}
\theoremstyle{definition}
\newtheorem{definition}[theorem]{Definition}
\newtheorem{conjecture}[theorem]{Conjecture}
\newtheorem{condition}[theorem]{Condition}
\setlist{nosep}
\title{Hardness Conjectures as Information Constraints\\[0.4em]{\large Does Complexity Theory Need New Axioms?}}
\author{Hunter Monroe}
\date{October 11, 2026}
\begin{document}
\maketitle
\begin{abstract}
We prove that no length-optimal propositional proof system exists if and only if every sound theory $\mathcal S{\supseteq}S^1_2$ with polynomial-time decidable axioms fails to simulate $S^1_2{+}\mathsf{Exp}{+}\phi_{BB}(k)$ for all sufficiently large $k$, where $\phi_{BB}(k)$ asserts the exact $k$-state Busy Beaver value and $\mathsf{Exp}$ asserts total binary exponentiation. Thus a widely believed proof-complexity conjecture admits an information-constraint formulation in terms of Busy Beaver values. For finitely axiomatized sequential $\mathcal S$, a Kolmogorov Blindness conjecture predicts exponential proof-length lower bounds for true randomness extensions, even after adjoining another true randomness axiom if the target remains random conditional on the added string. We also show that a separate hierarchy-level information constraint implies noncollapse of the polynomial hierarchy.

These information constraints parallel known results: Liu--Pass characterize one-way functions by average-case hardness of time-bounded Kolmogorov complexity $K^t$; Allender et al. relate $KT$ complexity to circuit size up to polynomial factors; Ball et al. characterize key agreement and one-way functions using one interactive problem with different gap parameters. We propose a research program to unify these connections through precise constraints on proof size, average-case time, and circuit size, and ask whether a common framework could justify a new axiom schema for complexity theory. Independence from $ZFC$ remains open. Unconditionally, for each fixed sound $\mathcal S$, any fixed sound computably axiomatized certifying theory extending $EA$ certifies simulation of only finitely many randomness extensions of $\mathcal S$. Defining efficiency as provable efficiency in a fixed theory relocates this difficulty.

\end{abstract}

\section{Introduction}
\begin{quote}
\emph{``\ldots the essence of the big open problems in complexity theory could be logical, rather than combinatorial.''}\hfill---Pudl\'ak~\cite{PudlakLogicalFoundations}
\end{quote}

Many central conjectures in complexity theory assert the absence of efficient procedures. An information-constraint perspective asks a more structural question: what mathematical information would such a procedure have to exploit, and what accounts for its access to that information? In proof complexity the question is concrete. A theory may verify individual finite computations without being able to justify uniformly the safety of adjoining a true statement. This motivates using the lengths of those bounded-consistency proofs to measure the theory's difficulty in accessing the information supplied by the additional sentence.

Information constraints describe the blind spots of formal systems: limits on the mathematical information that a fixed theory can certify or exploit through short proofs. These limitations are relative to the theory; information unavailable to one system may become accessible in a stronger one. We ask whether natural hardness conjectures reflect analogous limitations, and whether their persistence points to information principles absent from our current foundations. A common principle would gain support by identifying these limitations precisely and yielding further consequences beyond the conjectures it was introduced to explain. For instance, if, for almost all sufficiently long Kolmogorov-random strings $x\in \R$, a theory $\mathcal S$ cannot efficiently prove $\Con_{\mathcal S+(x\in \R)}(n)$ even at appropriately small cutoffs $n$, then the propositional translations of these bounded-consistency statements provide a dense family of small hard tautologies. Establishing such connections, and determining their independence from strong ambient theories, could provide an extrinsic case for new axioms in complexity theory.

The underlying quantitative question goes back to G\"odel, who suggested revisiting Hilbert's program through the lengths of proofs of finite consistency (Kreisel~\cite[p.~241]{Kreisel}). Pudl\'ak~\cite[p.~310]{Pudlak1987} formulates the general problem explicitly: given theories $\mathcal S$ and $\mathcal T$, how long is a shortest $\mathcal S$-proof of $\Con_{\mathcal T}(n)$? Kraj\'i\v{c}ek and Pudl\'ak~\cite{KP89} connect the existence of uniform polynomial upper bounds for such finite-consistency statements to the existence of an optimal propositional proof system, and their systematic failure to its nonexistence. Polynomial bounds offer a notion of practical reduction. We take this historical question as motivation for characterizing finite-consistency simulation through the information supplied by additional axioms.

We show that axioms stating Busy Beaver values provide canonical witnesses to non-simulation. For fixed $k$, the exact value $\BB(k)$ bounds the halting time of every halting $k$-state Turing machine. When a contradiction-search machine for a theory can be represented within this state bound, verifying that it has not halted by time $\BB(k)$ converts the Busy Beaver fact into a consistency assertion for that theory (Aaronson~\cite[Proposition 3]{AaronsonBusyBeaver}). Kolmogorov Blindness proposes randomness axioms as another source of hardness, even when a second true randomness axiom is available and the target remains random conditional on its string. A fixed sound effective theory can certify only finitely many sufficiently strong individual incompressibility assertions, although almost all long strings satisfy them. This limitation motivates Blindness but does not prove it. Hrube\v{s}'s self-reference construction refutes the earlier general relative-consistency characterization (Proposition~\ref{prop:hrubes-refutation}); the hardness conjecture here makes no inference from unprovability alone to proof-length lower bounds.

The ambition is to rule out \emph{inexplicable efficiency}: putative efficiency that appears to rely on hidden advice about Busy Beaver values or membership in $\R$. Here ``hidden advice'' describes a proposed informational explanation, not an assumed advice string or oracle in the computational model. The research task is to specify the required access and establish its necessity, with additional consequences testing the explanation. This raises a broader possibility: some of complexity theory's persistent open problems may reflect missing information constraints, potentially independent of the ambient foundational theory. Their formulation and consequences would then become part of the search for a resolution.

This viewpoint has substantial precedents outside proof theory. Time-bounded description complexity characterizes one-way functions through an average-case computational task. Bit-access description complexity is polynomially related to circuit size. Interactive description complexity places key agreement and one-way functions within one parameterized promise problem. These are mathematical relationships, not merely analogies. They suggest a positive program: identify a common measure whose resource-specific interpretations recover these conditions and whose additional structure yields new reductions or consequences.

\paragraph{Results and scope.}
Section~\ref{sec:bb} proves an unconditional consistency-extension characterization of nonexistence of optimal proof systems and an unconditional $S^1_2{+}\ExpAx$ Busy Beaver specialization. Section~\ref{sec:cert} proves unconditional certification theorems, with self-consistency and incompressibility consequences. Section~\ref{sec:kh} refutes the earlier characterization and develops Kolmogorov Blindness, its finite-scale density consequences, and a separate hierarchy condition. Section~\ref{sec:meta} states the three established metacomplexity correspondences with their different quantifiers. Sections~\ref{sec:program} and~\ref{sec:def} address, respectively, the axiom and independence program and the definitional approach. These latter topics are part of the mathematical body, not only motivation in the abstract.

Polynomial-size proofs can exist without a named base theory proving that they exist. Our certification theorem concerns the latter assertion; Kolmogorov Blindness conjectures lower bounds on the proofs themselves. The broader ambition is an alternative endgame for some complexity conjectures: find a natural information principle, establish its exact instances, and assess whether its explanatory and deductive value supports an axiom-like role.

\paragraph{Related work.}
Kraj\'i\v{c}ek and Pudl\'ak first connected optimal proof systems with polynomial-size finite-consistency proofs~\cite[Theorem 2.1]{KP89}; for the polynomial-time-axiomatized theory class used here, we use the formulation stated by Khaniki~\cite[Theorem 2.3]{Khaniki}. Pudl\'ak's survey places such statements in a finite-domain incompleteness program and asks whether their relations form one principle or a branching hierarchy~\cite{Pudlak}. The present program gives that question an information-theoretic focus and separates external proof bounds from their certification. Aaronson's Busy Beaver argument~\cite[Proposition 3]{AaronsonBusyBeaver} motivates the canonical information object. The metacomplexity characterizations below are due to Liu and Pass~\cite{LP}, Allender et al.~\cite{ABKMR}, and Ball et al~\cite{BLMP}. Our objective is to connect their resource interpretations, not to claim those characterizations as new.

Caldon and Ignjatovi\'c~\cite[Meta-definition 1 and equation (5)]{Caldon} also develop Kreisel's formulation, requiring the chosen theory to verify the existence of the short finite-consistency proofs. Thus the distinction between external feasibility and internally justified feasibility has an explicit precedent. Our certification analysis keeps the proving and certifying theories separate and identifies a relative-consistency consequence under specified syntactic assumptions.

The proof-theoretic core is shared with the revised companion paper \emph{Toward a Characterization of Simulation Between Arithmetic Theories}~\cite{MonroeCharacterizing}, which develops finite Busy Beaver transfer (Theorem~3.11) and a quantitative density--randomness equivalence (Theorem~6.1 and Corollary~6.3). The finite transfer takes an individual exponential lower bound at a sufficiently large consistency parameter $N$ to a Busy Beaver witness at a polynomially larger parameter, with a lower bound still exponential in $N$. For fixed $\Pi_1$ axiom families with computable proof-length bounds and a common computable onset, the density result equates polynomially small exceptional fractions with hardness above corresponding Kolmogorov-complexity thresholds; its pairwise form uses joint complexity. It does not establish the required density or a computable onset for Blindness. We use the companion's Kolmogorov Blindness formulation (Conjecture~5.1) and focus here on certification barriers and the proposed connections between resource models.

\paragraph{Organization.}
The paper is organized as follows. Section~\ref{sec:setup} fixes the proof-theoretic framework and the relevant information objects. Section~\ref{sec:bb} establishes the unconditional Busy Beaver characterization with exponentiation. Section~\ref{sec:cert} develops certified simulation and its relative-consistency consequences, while Section~\ref{sec:kh} refutes HRC and formulates Kolmogorov Blindness, finite-scale density, and a separate hierarchy-level condition. Section~\ref{sec:meta} reviews the established metacomplexity characterizations involving one-way functions, circuit size, and key agreement. Section~\ref{sec:program} develops the proposed common measure and evaluates it as a possible axiom schema, and Section~\ref{sec:def} analyzes the alternative notion of provable feasibility and its implications for certified algorithms and circuits. The concluding Section~\ref{sec:conclusion} collects the principal open questions, and Appendix~\ref{app:pairwise} derives pairwise hardness and factual nonaccess under the stated forms of Blindness.

\section{Preliminaries}\label{sec:setup}
Fix standard binary encodings of formulas and proofs. Use ordinary Hilbert proofs with reusable lines and cut (Pudl\'ak~\cite[Sections~2.4--2.5]{Pudlak}). Size is the length of the complete proof string, including formulas and line references. Polynomial-time axiom predicates use clocked verifiers. The ordinary proof predicates are represented by bounded formulas provably equivalent in $\Sbase$ to these verifier computations. A theory extending $\Sbase$ uses either literal axiom inclusion or a specified feasible embedding of $\Sbase$-proofs, with correctness provable in the base used in the argument. Theorem inclusion alone does not supply such an embedding. Consistency statements always use the fixed ordinary syntactic proof predicates. Write $\Con_{\mathcal U}(n)\equiv\forall\pi\bigl(|\pi|{\le}n\to
\neg\Prf_{\mathcal U}(\pi,\ulcorner 0=1\urcorner)\bigr)$ and $\Con_{\mathcal U}\equiv\forall n\,\Con_{\mathcal U}(n)$.

Numerals for $n$ are binary. We treat finite-consistency formulas as an effectively generated family indexed externally by $n$. Exponentially bounded proof-code assertions are formalized in $\EA$, where exponentiation is available. Weak bounded-arithmetic formalizations are specified separately. Write $\bot$ for $\ulcorner0{=}1\urcorner$.

Fix a bounded formula $E_2(x,y)$ for $y=2^x$ whose elementary defining properties are verified in $S^1_2$, including $E_2(x,y){\to}|y|=x+1$ under our binary convention, and write
$ \ExpAx\;\equiv\;\forall x\,\exists y\,E_2(x,y)$. The theory $\Sbase{+}\ExpAx$ remains finitely and polynomial-time axiomatized.

\begin{definition}[Finite-consistency simulation]\label{def:sim}
Write $\mathcal U{\csim}\mathcal S$ if there are standard $c,n_0$ such that, for every $n{\ge}n_0$, $\Con_{\mathcal U}(n)$ has an $\mathcal S$-proof (in the standard model) of size at most $n^c$. No polynomial-time algorithm producing the proofs is required.
\end{definition}

For Cook--Reckhow propositional proof systems, \emph{optimality} likewise means simulation with polynomial proof-length overhead; \emph{p-optimality} additionally requires an efficient translation. Throughout, optimality is size-only. Our arithmetic relation is explicitly Definition~\ref{def:sim}, not an unqualified identification with every associated propositional system. Its negation is
\[
\forall c\ \forall N\ \exists n{\ge}N:
\quad\text{no $\mathcal S$-proof of $\Con_{\mathcal U}(n)$ has size at most $n^c$.}
\tag{1}\label{eq:io}
\]
This is infinitely-often hardness in $n$, not eventual pointwise hardness.

The exponential rates below refer to the fixed proof encoding; polynomial proof translations preserve simulation but need not preserve these rates.

Fix a deterministic machine convention for the Busy Beaver running-time function, with effective padding by unreachable states. Let $\Gamma(k,t)$ be the effective arithmetic template asserting that every coded halting computation of a $k$-state machine first halts within $t$ steps and that some such machine attains the bound. For standard $k$, let $t_k{=}\BB(k)$ and $\phi_{\BB}(k){=}\Gamma(k,t_k)$. The attainment clause is a true fixed finite computation. The enormous numeral $t_k$ is fixed before the asymptotic proof-bound parameter $n$. The template $(k,t)\mapsto\Gamma(k,t)$ is effective; selecting $k\mapsto t_k$ is not.\footnote{An equivalent axiom may name the G\"odel number $e_k$ of a $k$-state runtime champion and assert that $e_k$ halts and has maximal first-halting time among $k$-state machines. For fixed standard $k$, $\Sbase$ verifies the actual finite trace of $e_k$ and, by finite unrolling of determinism, proves that every first-halting trace of $e_k$ has runtime $t_k$. Thus the champion axiom and $\Gamma(k,t_k)$ are interderivable over $\Sbase$, giving fixed inter-axiom proofs and polynomial proof translations in both directions. Constants may depend on $k$; no uniformly feasible selection of champions or runtimes is asserted. The machine code alone, without the halting and maximality assertions, is not an equivalent axiom.}

For randomness, fix an optimal plain description machine $U$ and a positive rational $d$. For $m{\ge}2$ put $\ell(m){=}\max\{0,\lceil m-d\log_2m\rceil\}$ and $\R{=}\{x:|x|{\ge}2,\ K_U(x){\ge}\ell(|x|)\}$. This is the companion paper's logarithmic-deficiency randomness set. Rational $d$ makes the integer threshold computable by integer comparisons. For all sufficiently large $m$, $2^{\ell(m)}{\le}2^{m+1}/m^d$, so the exceptional fraction is $O(m^{-d})$. Plain $K$, time-bounded $K^t$, and bit-access $\KT$ are distinct measures.

\section{Busy Beaver values and optimal systems}\label{sec:bb}
Busy Beaver values provide a canonical source of information about consistency. For any fixed theory $\mathcal T$, one can construct a machine that searches through $\mathcal T$-proofs and halts exactly when it finds a contradiction. Once this machine is represented within a sufficiently large state bound, the exact Busy Beaver value for that bound can be used to show that the search never halts and hence that $\mathcal T$ is consistent. We use this mechanism to replace theory-dependent consistency witnesses with a single natural family of information-bearing axioms and thereby connect their simulation hardness to the nonexistence of optimal propositional proof systems.
\begin{proposition}\label{prop:eaadequacy}
Let $\mathcal T$ be consistent and have a fixed polynomial-time ordinary proof predicate. For an explicitly compiled exhaustive $\mathcal T$-refutation search with $S^1_2$-verified padding, there is a standard $K^{\ExpAx}_{\mathcal T}$ such that, for every standard $k{\ge}K^{\ExpAx}_{\mathcal T}$,
$\Sbase{+}\ExpAx{+}\phi_{\BB}(k)\vdash\Con_{\mathcal T}$. No proof-length bound uniform in $k$ is asserted.
\end{proposition}
\begin{proof}
Let $E_{\mathcal T}$ enumerate candidate proof strings in length-lexicographic order, apply the clocked polynomial-time verifier, and halt upon finding a refutation. Use canonical tableaux whose bit length for a history of at most $t$ steps is bounded by a fixed polynomial $q_M(t)$, as verified in $\Sbase$. For fixed standard $t$, $\operatorname{HaltBy}_M(\bar t)$ is therefore equivalent in $\Sbase$ to a bounded search over $h{<}\overline{2^{q_M(t)+1}}$; the bound is a fixed numeral. Fix a polynomial $q$ bounding the verifier. From a supplied refutation $p$, the machine halts within a bound such as $2^{|p|+1}q(|p|)$, up to a fixed coding factor, and the complete tableau has size bounded by another fixed polynomial in this quantity. Choose a fixed polynomial $r$ for which $2^{r(|p|)}$ dominates the complete-tableau bound. There is a fixed polynomial-time function $F_{\mathcal T}$ which uses $z$ as padding and writes this tableau, and $S^1_2$ verifies
$\Prf_{\mathcal T}(p,\bot)\wedge E_2(r(|p|),a)\wedge E_2(a,z)
\ \longrightarrow\
\operatorname{HaltTrace}_{E_{\mathcal T}}\bigl(F_{\mathcal T}(p,z)\bigr)$.

Indeed, $E_2(a,z)$ gives $|z|{=}a{+}1$, so the running time and output length of $F_{\mathcal T}$ are polynomial in the length of its padded input. Let $F_{\mathcal T}$ simulate at most $|z|$ steps, returning a dummy string if the search has not halted. Bounded length induction verifies its history and search invariant in $\Sbase$; the supplied refutation ensures halting within this bound. In $\Sbase{+}\ExpAx$, two applications of total exponentiation supply the required $a$ and $z$. Hence
$\forall p\bigl(\Prf_{\mathcal T}(p,\bot)\to\exists h\,\operatorname{HaltTrace}_{E_{\mathcal T}}(h)\bigr)$.
Pad $E_{\mathcal T}$ by unreachable states. For every sufficiently large standard $k$, let $E_{\mathcal T,k}$ be the resulting $k$-state machine and put $t_k=\BB(k)$. The stipulated coding gives the following three internal facts:
\[
\begin{aligned}
S^1_2&\vdash \operatorname{Halt}_{E_{\mathcal T}}\to\operatorname{Halt}_{E_{\mathcal T,k}},\\
S^1_2{+}\phi_{\BB}(k)&\vdash \operatorname{Halt}_{E_{\mathcal T,k}}\to\operatorname{HaltBy}_{E_{\mathcal T,k}}(t_k),\\
S^1_2&\vdash \neg\operatorname{HaltBy}_{E_{\mathcal T,k}}(t_k).
\end{aligned}
\]
The first line is the verified padding equivalence. The second is the upper-bound component of the Busy Beaver axiom. For the third, consistency of $\mathcal T$ implies that the actual first $t_k$ steps of $E_{\mathcal T,k}$ are nonhalting. The fixed-parameter history bound makes this equivalent in $\Sbase$ to a true closed bounded sentence, which has a finite $\Sbase$-proof. Its length need not be feasible. A hypothetical $\mathcal T$-refutation contradicts these three facts. Constants and the finite verification proof may depend arbitrarily on the fixed standard $k$.
\end{proof}

Proposition~\ref{prop:eaadequacy} explicitly adds $\ExpAx$; it does not establish the corresponding assertion over $S^1_2$ alone.

\begin{lemma}\label{lem:transport}
Suppose a polynomial-time transformation sends every $\mathcal T$-refutation of size at most $n$ to a $\mathcal U$-refutation of size at most $p(n)$ for a fixed polynomial $p$, and its correctness is provable in $\Sbase$. Replacing $p$ by a monotone polynomial envelope, assume $p(n){\ge}n$. Then $\Con_{\mathcal U}(p(n))\to\Con_{\mathcal T}(n)$ has polynomial-size $\Sbase$-proofs. For every feasibly presented $\mathcal S{\supseteq}\Sbase$, $\mathcal U\csim\mathcal S$ therefore implies $\mathcal T\csim\mathcal S$.
\end{lemma}
\begin{proof}
Contrapose the formalized correctness implication and substitute the binary numeral for $n$. Evaluating the fixed polynomial has polynomial-size proofs. An $\mathcal S$-proof of $\Con_{\mathcal U}(p(n))$ of size $p(n)^c$, followed by this implication, has size polynomial in $n$. Selecting the original short proofs need not be effective.
\end{proof}

\begin{lemma}\label{lem:contransport}
If $\mathcal U{\supseteq}\Sbase$ has a fixed proof of $\Con_{\mathcal T}$ and the stipulated feasible embedding, the hypotheses of Lemma~\ref{lem:transport} hold.
\end{lemma}
\begin{proof}
A supplied $\mathcal T$-refutation $\pi$ has a polynomial-time accepting verifier computation. Write its bounded-arithmetic verification, yielding a proof of $\neg\Con_{\mathcal T}$, embed this in $\mathcal U$, append its fixed consistency proof, and derive contradiction. This has polynomial overhead in $|\pi|$ and formally checkable correctness. It does not require a transcript of exhaustive proof search.
\end{proof}

\subsection{Positive simulation and consistency-extension normal form}
The same transport mechanism explains a standard source of positive simulations. This supplies a benchmark for the proposed converse: a lower-bound principle must distinguish genuinely unexplained extensions from extensions admitting an efficient proof translation.

\begin{proposition}\label{prop:positive}
Let $\mathcal S,\mathcal U{\supseteq}\Sbase$ have the stipulated feasible presentations, and suppose $\mathcal S\csim\mathcal S$. If an $\Sbase$-verified polynomial refutation translation maps $\mathcal U$ to $\mathcal S$, then $\mathcal U\csim\mathcal S$. The conclusion also holds if $\Sbase\vdash\Con_{\mathcal S}{\to}\Con_{\mathcal U}$. More generally, $\mathcal U\csim\mathcal S$ is equivalent to the existence, for all sufficiently large $n$, of polynomial-size $\mathcal S$-proofs of $\Con_{\mathcal S}(p(n))\to\Con_{\mathcal U}(n)$ for some fixed polynomial $p(n){\ge}n$.
\end{proposition}
\begin{proof}
The first assertion is Lemma~\ref{lem:transport} with target $\mathcal S$.

For the unbounded implication, classical logic gives $\Sbase\vdash\forall\pi\,\exists\rho\,(\neg\Prf_{\mathcal U}(\pi,\bot)\vee\Prf_{\mathcal S}(\rho,\bot))$. The matrix is bounded. Parikh's theorem (Buss~\cite[Section~4.7, Theorem~11, p.~83]{Buss1986bounded}) supplies a term bound $\rho{\le}t(\pi)$. The binary length of a bounded-arithmetic term is polynomially bounded in its input lengths. Hence some polynomial $p(n){\ge}n$ satisfies $\Sbase\vdash\forall n\,(\Con_{\mathcal S}(p(n)){\to}\Con_{\mathcal U}(n))$. Instantiating this fixed proof gives polynomial-size $\mathcal S$-proofs of the implications. Combining them with the assumed short proofs of $\Con_{\mathcal S}(p(n))$ proves simulation.

For the general assertion, combine the displayed short implications with the assumed polynomial-size proofs of $\Con_{\mathcal S}(p(n))$. Conversely, from a short proof of $\Con_{\mathcal U}(n)$ derive $\Con_{\mathcal S}(p(n))\to\Con_{\mathcal U}(n)$ by propositional weakening, for any fixed $p$. All overheads remain polynomial.
\end{proof}

A polynomial-time interpretation with the base-verified refutation translation of Je\v{r}\'abek, presented by Pudl\'ak~\cite[Lemma~3.6]{Pudlak}, in its finitely axiomatized sequential setting, is an important instance. Notice the separate self-consistency hypothesis: a syntactic interpretation alone does not supply it in this argument. Also, a family of short proofs of implications is not the assertion that $\EA$ proves an unbounded relative-consistency implication. Keeping these levels separate is essential to the converse question. The Parikh argument uses $\Sbase$; sufficiency of an ordinary $\EA$ relative-consistency implication remains open here.

The known Kraj\'i\v{c}ek--Pudl\'ak characterization says that an optimal system exists exactly when some consistent polynomial-time axiomatized $\mathcal S{\supseteq}\Sbase$ simulates every consistent polynomial-time axiomatized $\mathcal T{\supseteq}\Sbase$ in the finite-consistency sense~\cite[Theorem 2.1]{KP89}\cite[Theorem 2.3]{Khaniki}. The polynomial may depend on $\mathcal T$. The next construction supplies soundness in the optimal-system direction.

\begin{lemma}\label{lem:cookpullback}
For every polynomial-time axiomatized $\mathcal T{\supseteq}\Sbase$, there is a family $\tau_{\mathcal T,n}$, polynomial-time generated from $1^n$ and of size polynomial in $n$, such that $\Sbase$ has polynomial-size proofs of
\[
\operatorname{Taut}(\tau_{\mathcal T,n})\to\Con_{\mathcal T}(n).
\]
If $\mathcal T$ is consistent, every $\tau_{\mathcal T,n}$ is a tautology.
\end{lemma}
\begin{proof}
Let $\tau_{\mathcal T,n}$ negate a Cook encoding of an accepting computation of the clocked $\mathcal T$-refutation verifier on an input of length at most $n$. Fixed polynomial-time maps on $(1^n,p)$ construct the Cook formula and, from a purported $\mathcal T$-refutation $p$, an accepting computation tableau and therefore an assignment falsifying $\tau_{\mathcal T,n}$. The fixed clocking makes these maps polynomial in $n$, and $S^1_2$ verifies their correctness with $n^{O(1)}$-size proofs. This is the computation-verification and finite-consistency translation used in the Kraj\'i\v{c}ek--Pudl\'ak correspondence~\cite[Lemmas 1.1, 3.2 and 3.3]{KP89}. Consistency gives tautologicity.
\end{proof}

\begin{lemma}\label{lem:sound}
If an optimal propositional proof system exists, some sound polynomial-time axiomatized $\mathcal S{\supseteq}\Sbase$ satisfies $\mathcal T\csim\mathcal S$ for every consistent polynomial-time axiomatized $\mathcal T{\supseteq}\Sbase$.
\end{lemma}
\begin{proof}
Let $P$ be optimal. Adjoin to $\Sbase$ the true arithmetic sentence $\operatorname{RFN}(P)$ saying every output of $P$ is a tautology; call the sound result $\mathcal S_P$. For fixed consistent $\mathcal T$, take the tautologies $\tau_{\mathcal T,n}$ from Lemma~\ref{lem:cookpullback}.

A proof system can output $\tau_{\mathcal T,n}$ on specially prefixed unary proofs $1^n$, with a separate prefix supporting a complete propositional system. Optimality gives polynomial-size $P$-proofs of these tautologies. For each such proof, $\mathcal S_P$ verifies the computation of $P$, applies $\operatorname{RFN}(P)$, and then the encoding implication, obtaining a polynomial-size proof of $\Con_{\mathcal T}(n)$. No proof of optimality inside $\mathcal S_P$, or uniform selection of the short $P$-proofs, is used.
\end{proof}

\begin{theorem}\label{thm:normal}
No optimal propositional proof system exists if and only if, for every sound polynomial-time axiomatized $\mathcal S{\supseteq}\Sbase$, there is a consistent polynomial-time axiomatized $\mathcal T{\supseteq}\Sbase$ such that
\[
\mathcal S{+}\Con_{\mathcal T}\not\csim\mathcal S.
\]
In particular, the added axiom is true, even though $\mathcal T$ need not be sound.
\end{theorem}
\begin{proof}
If no optimal system exists, the known characterization gives, for each fixed sound $\mathcal S$, a consistent $\mathcal T$ with $\mathcal T\not\csim\mathcal S$. Put $\mathcal V{=}\mathcal S{+}\Con_{\mathcal T}$. It is sound because $\Con_{\mathcal T}$ is true. Lemmas~\ref{lem:contransport} and~\ref{lem:transport} show that $\mathcal V\csim\mathcal S$ would imply $\mathcal T\csim\mathcal S$, a contradiction. Conversely, if an optimal system exists, choose the sound $\mathcal S_P$ from Lemma~\ref{lem:sound}. Every $\mathcal S_P{+}\Con_{\mathcal T}$ for consistent $\mathcal T$ is a consistent polynomial-time axiomatized target and is therefore simulated by $\mathcal S_P$. This contradicts the asserted universal property.
\end{proof}

This normal form is unconditional. It explains why consistency information is already sufficient for a characterization. The next theorem shows that, after adjoining $\ExpAx$, a fixed, independently defined Busy Beaver family supplies those witnesses.

\subsection{Canonical Busy Beaver witnesses}
\begin{theorem}\label{thm:jointbase}
No optimal propositional proof system exists if and only if, for every sound polynomial-time axiomatized $\mathcal S{\supseteq}S^1_2$, some $K_{\mathcal S}$ satisfies
\[
S^1_2{+}\ExpAx{+}\phi_{\BB}(k)\not\csim\mathcal S
\qquad\text{for every standard }k{\ge}K_{\mathcal S}.
\]
\end{theorem}
\begin{proof}
Assume there is no optimal system and fix $\mathcal S$. The Kraj\'i\v{c}ek--Pudl\'ak characterization supplies a consistent polynomial-time axiomatized $\mathcal T{\supseteq}S^1_2$ with $\mathcal T\not\csim\mathcal S$. By Proposition~\ref{prop:eaadequacy}, for every sufficiently large $k$ the theory $S^1_2{+}\ExpAx{+}\phi_{\BB}(k)$ proves $\Con_{\mathcal T}$. Lemmas~\ref{lem:contransport} and~\ref{lem:transport} therefore show that its simulation by $\mathcal S$ would imply $\mathcal T\csim\mathcal S$, a contradiction.

Conversely, suppose $P$ is optimal. Lemma~\ref{lem:sound} gives a sound polynomial-time axiomatized $\mathcal S_P{\supseteq}S^1_2$ that simulates every consistent polynomial-time axiomatized extension of $S^1_2$. Every $S^1_2{+}\ExpAx{+}\phi_{\BB}(k)$ is such an extension, so it is simulated by $\mathcal S_P$ for every $k$, contradicting the displayed property.
\end{proof}

\begin{corollary}\label{cor:bbtaut}
Assume that no optimal system exists. For every Cook--Reckhow system $P$, some $K_P$ satisfies the following: for every fixed $k{\ge}K_P$, the Cook translations $\tau_{\mathcal U_k,n}$ of finite consistency for $\mathcal U_k{=}\Sbase{+}\ExpAx{+}\phi_{\BB}(k)$ have no polynomial-size $P$-proof family.
\end{corollary}
\begin{proof}
Use the sound theory $\mathcal S_P{=}\Sbase{+}\operatorname{RFN}(P)$, without assuming $P$ optimal. The verification and Cook pullback argument in Lemma~\ref{lem:sound} shows that polynomial-size $P$-proofs of these translations would give $\mathcal U_k\csim\mathcal S_P$. Theorem~\ref{thm:jointbase} excludes this for all sufficiently large $k$. Each translation is a tautology because $\mathcal U_k$ is sound.
\end{proof}

For fixed $k$ the formula family is effectively generated in time polynomial in $n$; neither the threshold $K_P$ nor the true numeral $\BB(k)$ is asserted to be computable uniformly from $P,k$. The universality is over proof systems, not a claim of one efficiently selectable family that defeats them all.

\section{Certified simulation and relative consistency}\label{sec:cert}
External simulation asserts that the required short proofs exist, but not that any specified theory can prove their existence. This section studies the stronger requirement that a base theory $\mathcal B$ certify a simulation. We show that such certification entails a relative-consistency implication over $\mathcal B$, producing G\"odelian barriers to what a theory can verify without by itself establishing external proof-length lower bounds. Throughout, $\mathcal B$ extends $\EA$, and the chosen presentations admit $\EA$-provably correct elementary operations for numeral substitution, witness verification, proof embedding, and proof concatenation. These are syntactic requirements, not assumptions of soundness.

For standard $c,n_0$ put
$\Sim^{c,n_0}_{\mathcal S,\mathcal U}\equiv
\forall n{\ge}n_0\ \exists\pi
\bigl(|\pi|{\le}n^c\land
\Prf_{\mathcal S}(\pi,\ulcorner\Con_{\mathcal U}(\bar n)\urcorner)\bigr)$. The quotation denotes the elementary numeral-substitution operation. Removing the size bound gives $\Sim^{n_0}_{\mathcal S,\mathcal U}$. A base \emph{certifies} a simulation when it proves one of these sentences.

\paragraph{Quantifiers and logical complexity.}
For fixed $c,n_0$, the proof code in the bounded sentence ranges below an elementary bound $2^{n^c+1}$, so the remaining matrix is decidable by a bounded computation. The fixed-bound simulation sentence is $\Pi^0_1$. Existence of some polynomial bound, $\exists c,n_0\,\Sim^{c,n_0}_{\mathcal S,\mathcal U}$, is $\Sigma^0_2$; the unbounded simulation sentence is $\Pi^0_2$. These are logical forms, not completeness claims. The truth of a fixed-bound sentence is precisely an external short-proof-family assertion. Proving it in a named base is an additional requirement, not a choice of notation for the same requirement.

\begin{lemma}\label{lem:witness}
For the stipulated presentations, $\EA$ proves that a $\mathcal U$-refutation $\rho$ of length at most $n$ supplies an $\mathcal S$-proof of $\neg\Con_{\mathcal U}(n)$. It also verifies the combination of opposite $\mathcal S$-proofs into a refutation.
\end{lemma}
\begin{proof}
Evaluate the clocked verifier on the supplied $\rho$, write a proof of that computation, introduce its existential witness, and embed in $\mathcal S$. For the second operation concatenate the proofs and append a propositional contradiction. Elementary induction verifies these constructions in $\EA$; no inclusion $\EA{\subseteq}\mathcal S$ is needed.
\end{proof}

\begin{theorem}\label{thm:cert}
For $\mathcal B{\supseteq}\EA$ and the preceding presentations,
\[
\mathcal B\vdash\Sim^{n_0}_{\mathcal S,\mathcal U}
\quad\Longrightarrow\quad
\mathcal B\vdash\Con_{\mathcal S}\to\Con_{\mathcal U}.
\]
The conclusion holds in particular for certified polynomial simulation.
\end{theorem}
\begin{proof}
Inside $\mathcal B$, assume $\Con_{\mathcal S}$ and take a putative $\mathcal U$-refutation $\rho$. Set $n{=}\max\{|\rho|,n_0\}$. Lemma~\ref{lem:witness} gives an $\mathcal S$-proof of $\neg\Con_{\mathcal U}(n)$; certified simulation supplies an $\mathcal S$-proof of its negation at the same $n$. Combining these contradicts $\Con_{\mathcal S}$. Soundness is unnecessary, and the proof-size bound is unused.
\end{proof}

\begin{corollary}\label{cor:self}
Let $\mathcal S$ be a sound polynomial-time axiomatized extension of $\EA$ and $\Sbase$, with the standard presentations. Put $\mathcal V{=}\mathcal S{+}\Con_{\mathcal S}$. Then $\mathcal V\nvdash\Sim^{n_0}_{\mathcal S,\mathcal V}$ for every standard $n_0$, and neither $\mathcal V$ nor $\mathcal S$ certifies a polynomial simulation.
\end{corollary}
\begin{proof}
A certificate in $\mathcal V$ would imply $\mathcal V\vdash\Con_{\mathcal S}\to\Con_{\mathcal V}$ by Theorem~\ref{thm:cert}. But $\mathcal V$ contains $\Con_{\mathcal S}$, so this consistent effective theory would prove its own ordinary consistency, contrary to second incompleteness.
\end{proof}

\begin{lemma}\label{lem:random}
For fixed $x$ of length at least two and the stated presentations, $\EA\vdash\Con_{\mathcal S+(x\in\R)}\to x\in\R$.
\end{lemma}
\begin{proof}
A witness to $x\notin\R$ is a short program with a supplied halting history producing $x$. Elementary syntax constructs an $\mathcal S$-proof of $x\notin\R$. Combining it with the adjoined axiom gives a refutation. Contrapose this verified construction in $\EA$.
\end{proof}

\begin{lemma}[Chaitin~\cite{Chaitin}]\label{lem:chaitin}
Every fixed sound computably axiomatized theory $\mathcal A$ proves $x\in\R$ for only finitely many strings $x$.
\end{lemma}
\begin{proof}
Otherwise, on input $a$, enumerate proofs until finding $x\in\R$ with $\ell(|x|)>a$. Unbounded lengths ensure termination. A self-delimiting description of $a$ and the fixed search algorithm describe the resulting $x$ with $O(\log a)$ bits, contradicting its soundly asserted complexity greater than $a$ for sufficiently large $a$. Bounded lengths contain only finitely many strings.
\end{proof}

\begin{theorem}\label{thm:certkh}
Fix sound $\mathcal S{\supseteq}\Sbase$ with the stipulated presentation and a sound computably axiomatized $\mathcal B{\supseteq}\EA$. Only finitely many true $x\in\R$ have $\mathcal B\vdash\Sim^{c,n_0}_{\mathcal S,\mathcal S+(x\in\R)}$ for any standard $c,n_0$.
\end{theorem}
\begin{proof}
Each certificate gives $\mathcal B\vdash\Con_{\mathcal S}\to\Con_{\mathcal S+(x\in\R)}$, hence $\mathcal B{+}\Con_{\mathcal S}\vdash x\in\R$ by Lemma~\ref{lem:random}. This is one fixed sound effective theory, independent of $x,c,n_0$. Apply Lemma~\ref{lem:chaitin}.
\end{proof}

The finite exceptional set depends on the certifying base. The theorem does not rule out an externally existing short-proof family; it limits what each fixed sound base can certify, even if that base is very strong. Enlarging the base can enlarge its finite exceptional set, but cannot remove the finite-certification barrier while the base remains fixed, sound, and effective.

\begin{theorem}\label{thm:certbb}
Fix a sound polynomial-time axiomatized $\mathcal S{\supseteq}\Sbase$ with the stipulated feasible presentation and a fixed sound computably axiomatized $\mathcal B{\supseteq}\EA$. Only finitely many standard $k$ satisfy
\[
\mathcal B\vdash\Sim^{n_0}_{\mathcal S,\Sbase{+}\phi_{\BB}(k)}
\quad\text{for some standard }n_0.
\]
Consequently only finitely many $k$ admit a certified polynomial simulation. This result is independent of the exponentiation-base characterization in Section~\ref{sec:bb}.
\end{theorem}
\begin{proof}
Put $\mathcal A{=}\mathcal B{+}\Con_{\mathcal S}$, a fixed sound effective theory. Each certificate gives $\mathcal A\vdash\Con_{\Sbase{+}\Gamma(k,t_k)}$ by Theorem~\ref{thm:cert}. Enumerate all pairs $(k,t)$ for which $\mathcal A$ proves $\Con_{\Sbase{+}\Gamma(k,t)}$, dovetailing theorem enumeration with the effective sentence template.

Every enumerated pair is a correct upper bound: $t{\ge}\BB(k)$. Indeed, soundness of $\mathcal A$ makes $\Sbase{+}\Gamma(k,t)$ consistent. If $t<\BB(k)$, an actual finite computation first halts after $t$ steps; its particular trace has an $\Sbase$ verification refuting the upper-bound clause of $\Gamma(k,t)$. This would make that extension inconsistent. Only instancewise finite verification is used in this step.

If certified state sizes were unbounded, this enumeration would yield correct upper bounds at unbounded state sizes. Given any blank-input machine with $m$ states, wait for an enumerated pair $(k,t)$ with $k{\ge}m$, pad the machine to $k$ states, and run it for $t$ steps. The wait terminates and the result decides whether the machine ever halts, contradicting undecidability. Hence the set of certified state sizes is finite.
\end{proof}

This is the finite-information obstruction underlying the Busy Beaver incompleteness argument~\cite[Proposition 3]{AaronsonBusyBeaver}, applied to certification through Theorem~\ref{thm:cert}. It provides a theorem-level information constraint even while the external simulation question is kept separate. The proof does not turn arbitrary refutations into exhaustive-search histories inside $\Sbase$.

\subsection{Certification of feasible relative consistency}
The positive mechanism of Proposition~\ref{prop:positive} has a corresponding certification theorem. Here elementary, rather than polynomial, proof construction suffices.

\begin{lemma}\label{lem:finiteverify}
For the stipulated polynomial-time presentation of $\mathcal S{\supseteq}\Sbase$, an elementary function $h$ satisfies
\[
\EA\vdash\forall m\bigl(\Con_{\mathcal S}(m)\to
\Prf_{\mathcal S}(h(m),\ulcorner\Con_{\mathcal S}(\bar m)\urcorner)\bigr).
\]
No polynomial bound on the length of $h(m)$ is asserted.
\end{lemma}
\begin{proof}
Enumerate all binary strings of length at most $m$ and check the clocked refutation predicate. If a check accepts, output an arbitrary string. Otherwise, assemble the proofs of the failed checks and a finite case analysis over all candidate codes. This proves $\Con_{\mathcal S}(m)$ in $\Sbase$ and hence in $\mathcal S$ by the stipulated embedding. The number of checks, their individual verification proofs, the case analysis, and the numerical code of the assembled proof all have elementary bounds in $m$. Bounded induction with exponentiation verifies the construction and its conditional correctness in $\EA$.
\end{proof}

\begin{theorem}\label{thm:certfrc}
Retain the presentation assumptions of this section, including $\mathcal S{\supseteq}\Sbase$ and the $\EA$-verified syntactic maps. Let $\mathcal B{\supseteq}\EA$ and let $p$ be a fixed polynomial. Suppose there is a standard $n_0$ such that
\[
\mathcal B\vdash\forall n{\ge}n_0\ \exists q\,
\Prf_{\mathcal S}\bigl(q,
\ulcorner\Con_{\mathcal S}(\overline{p(n)})\to\Con_{\mathcal U}(\bar n)\urcorner\bigr).
\]
Then $\mathcal B\vdash\Con_{\mathcal S}\to\Con_{\mathcal U}$. This is one uniform proof-existence assertion, not separate proofs for each standard $n$. No bound on the asserted proof sizes is required.
\end{theorem}
\begin{proof}
Reason in $\mathcal B$ under $\Con_{\mathcal S}$. A putative $\mathcal U$-refutation $\rho$ gives an $\mathcal S$-proof of $\neg\Con_{\mathcal U}(n)$ at $n{=}\max\{|\rho|,n_0\}$. The certified implication gives an $\mathcal S$-proof of $\neg\Con_{\mathcal S}(p(n))$. But $\Con_{\mathcal S}$ implies $\Con_{\mathcal S}(p(n))$, so Lemma~\ref{lem:finiteverify} supplies an $\mathcal S$-proof of the latter. Concatenation contradicts $\Con_{\mathcal S}$.
\end{proof}

These results locate the difference between a uniform metatheorem and a family of individual finite verifications. Every true finite-consistency instance can be verified by a sufficiently long finite argument. A certificate that all required proofs exist simultaneously has a relative-consistency consequence; this is why the uniform assertion can remain unprovable even when each standard instance has a proof.
\section{Kolmogorov hardness and finite-scale consequences}\label{sec:kh}
Restrict $\mathcal S$ to sound finitely axiomatized sequential theories with the specified feasible presentation extending $\Sbase$. This is a restriction on the conjectures in this section, not on Theorem~\ref{thm:certkh}. An earlier version of this paper proposed Higher Relative Consistency (HRC), the characterization
\[
\mathcal S{+}\phi\csim\mathcal S
\quad\Longleftrightarrow\quad
\EA\vdash\Con_{\mathcal S}\to\Con_{\mathcal S{+}\phi}
\tag{2}\label{eq:characterization}
\]
for every such $\mathcal S$ and every true arithmetic sentence $\phi$. Its necessary direction is false.

\begin{proposition}[Refutation of HRC]\label{prop:hrubes-refutation}
There are a sound finitely axiomatized sequential $\mathcal S{\supseteq}\EA,\Sbase$ and a true $\Pi_1$ sentence $\rho$ such that $\mathcal S{+}\rho\csim\mathcal S$, but $\EA\nvdash\Con_{\mathcal S}{\to}\Con_{\mathcal S{+}\rho}$. Neither $\Sbase$ nor $\mathcal S$ proves this implication.
\end{proposition}
\begin{proof}
Choose such an $\mathcal S$, for example a finite presentation of $I\Sigma_1$ in the common definitional language with the fixed finite base axioms adjoined. Set $\mathcal T{:=}\mathcal S{+}\Con_{\mathcal S}$, which is sound and polynomial-time axiomatized. Hrube\v{s}'s theorem, presented by Pudl\'ak~\cite[Theorem~3.4]{Pudlak}, supplies a true $\Pi_1$ sentence $\rho$ with $\mathcal T\nvdash\rho$, although $\mathcal S{+}\rho\csim\mathcal S$.

Since $\neg\rho$ is $\Sigma_1$, formalized $\Sigma_1$-completeness gives $\EA\vdash\neg\rho{\to}\operatorname{Prov}_{\mathcal S}(\ulcorner\neg\rho\urcorner)$ (Artemov and Beklemishev~\cite[Proposition~16]{ArtemovBeklemishev2005}). Combining such a proof with the axiom $\rho$ yields a refutation of $\mathcal S{+}\rho$. Thus $\EA\vdash\Con_{\mathcal S{+}\rho}{\to}\rho$. An $\EA$-proof of the relative-consistency implication would therefore give $\mathcal T\vdash\rho$, since $\mathcal T$ contains $\EA$ and $\Con_{\mathcal S}$, a contradiction. The same argument works for a certificate in $\Sbase$ or $\mathcal S$, both contained in $\mathcal T$.
\end{proof}

Thus simulation does not require the certificate in (2), even with $\EA$ replaced by $\Sbase$ or by the simulating theory itself. Proposition~\ref{prop:positive} gives the surviving sufficient condition using $\Sbase$-provable relative consistency. The analogous sufficiency assertion for an arbitrary ordinary $\EA$ implication is not established here. For randomness extensions we instead propose the following direct hardness conjecture, as in the revised companion paper~\cite[Conjecture~5.1]{MonroeCharacterizing}.

Fix a universal conditional description machine compatible with $U$. For $|x|{=}|y|{=}m$, say that $x$ remains random conditional on $y$ when $K_U(x\mid y){\ge}\ell(m)$. Write $x\perp_{\log}y$ when $x,y\in\R$ and both conditional bounds hold. For the empty string $\varepsilon$, set $K_U(x\mid\varepsilon){:=}K_U(x)$ and let $\Delta(\varepsilon)$ add no axiom; for a nonempty $y$, let $\Delta(y)$ be $y\in\R$.

\begin{conjecture}[Kolmogorov Blindness]\label{conj:kolmogorov-blindness}
For every sound finitely axiomatized sequential $\mathcal S{\supseteq}\Sbase$ with its fixed proof presentation, there are a positive rational $\eta_{\mathcal S}$ and a cutoff $k_R^{\mathcal S}$ such that for every $m{>}k_R^{\mathcal S}$ there is $N_{R,m}^{\mathcal S}$ with the following property. For every true $x\in\R\cap\{0,1\}^m$, every $y$ that is either $\varepsilon$ or a true member of $\R\cap\{0,1\}^m$, and every $n{>}N_{R,m}^{\mathcal S}$, an $\mathcal S{+}\Delta(y)$-proof of $\Con_{\mathcal S+(x\in\R)}(n)$ of size at most $2^{\eta_{\mathcal S}n}$ implies $K_U(x\mid y){<}\ell(m)$.
\end{conjecture}

The rate may depend on the proof encoding, $U$ and $d$. The onset is common to all strings and helpers of length $m$, but no small or computable bound for it is asserted. Bounds arbitrarily close to $2^n$ would be a stronger encoding-specific hypothesis; Blindness does not supply them. Its cutoff $k_R^{\mathcal S}$ is a conjectural hardness threshold, distinct from the finite certification cutoff furnished by Chaitin incompleteness; no comparison between them is established. No implication from SETH for satisfiability is asserted.

\begin{theorem}[Single-axiom consequences of Blindness]\label{thm:kh}
Assume Kolmogorov Blindness at fixed $\mathcal S$. For every true $x\in\R$ with $m{:=}|x|{>}k_R^{\mathcal S}$ and every $n{>}N_{R,m}^{\mathcal S}$, $\Con_{\mathcal S+(x\in\R)}(n)$ has no $\mathcal S$-proof of size at most $2^{\eta_{\mathcal S}n}$. In particular, $\mathcal S+(x\in\R)\not\csim\mathcal S$ for every such $x$.
\end{theorem}
\begin{proof}
Take $y{=}\varepsilon$ in Conjecture~\ref{conj:kolmogorov-blindness}. A proof within the bound would give $K_U(x){<}\ell(m)$, contradicting membership in $\R$. Since $\eta_{\mathcal S}{>}0$, the eventual lower bound exceeds every polynomial in $n$ for this fixed $x$.
\end{proof}

\paragraph{Robustness of the explanatory boundary.}
Certification gives a useful conditional comparison. Suppose an explanation relation $E_{\mathcal S}(x)$ is necessary for simulation of every true randomness extension, and one fixed sound effective theory $\mathcal A$ proves $x\in\R$ whenever $E_{\mathcal S}(x)$ holds. Then the simulated true randomness extensions form a finite set, by Lemma~\ref{lem:chaitin}. Both hypotheses must be verified for the proposed mechanism; this observation does not supply them automatically or make them consequences of Blindness.

\paragraph{Fixed time-bounded facts.}
Every fixed true time-bounded randomness assertion is a finite computation statement and has a proof in the arithmetic base. Replacing that fixed axiom by its fixed proof gives polynomial refutation transport; if $\mathcal S\csim\mathcal S$, Proposition~\ref{prop:positive} supplies simulation of the extension. The fixed proof can be enormous and need not be uniformly short in the string length. This is compatible with average-case hardness of computing $K^t$ on varying inputs, but precludes substituting these fixed facts for unbounded randomness axioms in Blindness.

\begin{definition}[Strong KH and Density KH]\label{def:kh-properties}
For the fixed $\mathcal S,U,d$, \emph{Strong KH} holds if every sufficiently long true $x\in\R$ gives $\mathcal S+(x\in\R)\not\csim\mathcal S$. \emph{Density KH} holds if the fraction of length-$m$ strings that are true members of $\R$ and give such non-simulation tends to one as $m{\to}\infty$.
\end{definition}

\begin{theorem}[Dense finite-scale hardness]\label{thm:density}
Under Kolmogorov Blindness at fixed $\mathcal S$, for every sufficiently large $m{>}k_R^{\mathcal S}$, at least a $1-O(m^{-d})$ fraction of length-$m$ indices yield true finite-consistency statements requiring $\mathcal S$-proof size greater than $2^{\eta_{\mathcal S}n}$ for every $n{>}N_{R,m}^{\mathcal S}$. Strong KH and Density KH both hold.
\end{theorem}
\begin{proof}
Fewer than $2^{\ell(m)}$ programs have length below $\ell(m)$, and each outputs at most one string. The exceptional fraction is therefore $O(m^{-d})$. Apply Theorem~\ref{thm:kh} to every remaining string. That theorem gives Strong KH; the same counting estimate gives Density KH.
\end{proof}

\begin{corollary}[Hardness of consistency tautologies]\label{cor:taut}
Assume Kolmogorov Blindness at $\mathcal S$ and a translation $\tau_{x,n}$ of $\Con_{\mathcal S+(x\in\R)}(n)$ into formulas of size $(m+n)^{O(1)}$. Suppose every proof of $\tau_{x,n}$ in a fixed propositional system $P$ of length $s$ yields an $\mathcal S$-proof of length at most $s(m+n)^a$, for fixed $a$. For every fixed sufficiently large $m{>}k_R^{\mathcal S}$ and every $\beta$ with $0{<}\beta{<}\eta_{\mathcal S}$, a $1-O(m^{-d})$ fraction of the indices yield tautologies with no $P$-proof of size at most $2^{\beta n}$ for all sufficiently large $n$.
\end{corollary}
\begin{proof}
For fixed $m$, eventually $(m+n)^a{<}2^{(\eta_{\mathcal S}-\beta)n}$. A $P$-proof within the claimed bound would give an $\mathcal S$-proof smaller than $2^{\eta_{\mathcal S}n}$, contradicting Theorem~\ref{thm:density}. Truth of finite consistency for random $x$ gives tautologicity under the stipulated translation.
\end{proof}

The density is over indices, not all formulas. A distinct-tautology count also needs an injective translation or bounded multiplicity. The onset has not been bounded as a function of $m$; one cannot substitute $m=n$ without checking that onset. Obtaining useful small-scale bounds remains a substantive strengthening. Appendix~\ref{app:pairwise} derives the corresponding density of pairs for which each lower bound still holds after adjoining the other randomness axiom. Factual nonaccess there uses universal Blindness, rather than just the conjecture at the original $\mathcal S$.

Universal Blindness also rules out a length-optimal proof system. If $P$ were optimal, Lemma~\ref{lem:sound} would give the sound finitely axiomatized sequential theory $\mathcal S_P{:=}\Sbase{+}\mathrm{RFN}(P)$ simulating every consistent polynomial-time axiomatized extension of $\Sbase$. Blindness at $\mathcal S_P$ instead gives true randomness extensions not simulated by it, by Theorem~\ref{thm:kh}. This contradiction uses the conjecture across all eligible finite bases; an instance at one fixed theory does not suffice.

\subsection{A hierarchy-level condition}
For each $i{\ge}1$, fix a $\Pi_i^p$-complete language $L_i$ and an oracle description machine $V_i$. For polynomial $t$, let $Q_i^t(x)$ assert that no program of length below $\ell(|x|)$ outputs $x$ via $V_i^{L_i}$ within $t(|x|)$ steps. The oracle is a polynomial-hierarchy language, not the set of all true arithmetical $\Pi_i$ sentences.

\begin{condition}[SETH-K-PH]\label{cond:ph}
For each $i{\ge}1$, some polynomial $t_i$ satisfies $Q_i^{t_i}\notin\Pclass^{\Pi_i^p}$.
\end{condition}
\begin{theorem}\label{thm:ph}
Condition~\ref{cond:ph} implies $\PH$ noncollapse.
\end{theorem}
\begin{proof}
The complement of $Q_i^t$ guesses a short program and simulates its bounded $L_i$-oracle computation. Thus $\neg Q_i^t\in\NP^{\Pi_i^p}{=}\Sigma_{i+1}^p$ and $Q_i^t\in\Pi_{i+1}^p$. Equality $\Pi_i^p{=}\Pi_{i+1}^p$ would put this language in $\Pclass^{\Pi_i^p}$, a contradiction. A finite collapse would imply such adjacent equalities.
\end{proof}

Deriving this separate computational condition from proof-size hypotheses is a unification objective, not a suppressed premise of the theorem. Although counting makes its YES set dense, accept-all then has small uniform error. Density is not the average-case hardness needed below.

\section{Established metacomplexity instances}\label{sec:meta}
The following results share an information-theoretic theme but retain distinct resource models and quantifiers.

\subsection{Average-case time and one-way functions}
Under the conventions of Liu and Pass~\cite{LP}, $K^t(x)$ is the minimum length of a program producing $x$ within $t(|x|)$ steps. Mild average-case hardness means that some polynomial $p$ lower-bounds the error probability by $1/p(n)$ for every probabilistic polynomial-time algorithm, on uniform $n$-bit inputs and its random coins, at all sufficiently large lengths.

\begin{theorem}[Liu--Pass]\label{thm:lp}
For polynomial $t(n){\ge}(1{+}\epsilon)n$ with fixed $\epsilon>0$, using the source's universal-machine conventions, one-way functions exist if and only if computing $K^t$ is mildly hard on average.
\end{theorem}

The task asks for time-bounded description length, not unbounded randomness certification. It admits a natural information reading: efficient procedures cannot reliably recover this description information from the input alone. Kolmogorov Blindness proposes a limit on using true randomness axioms through short finite-consistency proofs. A quantitative connection between these two forms of access would turn their common explanation into a transfer theorem. Neither Blindness nor the companion's effective density equivalence supplies that bridge.

\subsection{Circuit size and bit-access descriptions}
For a truth table $z$ of length $2^m$, $\KT(z)$ minimizes $|d|+T$, where a fixed universal machine with random access to description $d$ answers individual bit queries within time $T$, including the end-marker test at position $|z|+1$. Let $\operatorname{Size}(z)$ be the smallest circuit for its Boolean function. We count input gates as well as connections, as required by the clarification~\cite{Clarification}; in particular, $\operatorname{Size}(z){\ge}m$.

\begin{theorem}[Allender et al.]\label{thm:kt}
With these conventions, $\KT(z)+m$ and $\operatorname{Size}(z)$ are bounded by fixed polynomials in one another.
\end{theorem}

Writing $\kappa{=}\KT(z)$ and $s{=}\operatorname{Size}(z)$, the source gives $s\le C\kappa^2(\kappa^2+m)$ and $\kappa\le Cs^2(\log s{+}\log m)$, up to adjustment at finitely many small sizes~\cite[Theorem 11]{ABKMR}. Together with $s{\ge}m$, these imply the stated relation. Superpolynomial requirements in circuit input length $m$ therefore correspond under the two measures. The parameter is $m$, not truth-table length $2^m$. The relation concerns descriptive complexity of truth tables; it does not equate circuit lower bounds with hardness of deciding whether a given string has low $\KT$.

\subsection{Interaction and key agreement}
Ball et al.~\cite{BLMP} provide a concrete precedent for parameterized unification. Their $\mathsf{IK}^t(\pi;x;y)$ is the minimum total description length of two deterministic interacting programs, each with budget $t(n)$, alternating bit transmissions and producing transcript $\pi$ and outputs $x,y$. Conditional $K^t(x\mid y)$ uses $y$ as auxiliary input and budget $t(|x|)$; $K(\pi,x,y)$ uses a fixed self-delimiting tuple encoding. On triples of $n$-bit strings define the source's promise problem
\[
\begin{aligned}
Y&{=}\{(\pi,x,y):\mathsf{IK}^t(\pi;x;y)\le K(\pi)+10\log n\},\\
N&{=}\{(\pi,x,y):K(\pi,x,y)\ge K(\pi)+50\log n\},
\end{aligned}
\]
restricted to $Q_\Delta{=}\{(\pi,x,y):K^t(x\mid y),K^t(y\mid x)\le\Delta(n)\}$. Put $J_\Delta=(Y\cap Q_\Delta,N\cap Q_\Delta)$, discarding finitely many small lengths so that the two promise sides are disjoint, as in the source.

\begin{theorem}[Ball--Liu--Mazor--Pass, two instances]\label{thm:blmp}
For polynomial $t(n)>n^{1.1}$ and the source's conventions,
\[
\begin{aligned}
\text{key agreement exists}&\ \Longleftrightarrow\ J_{\log n}\notin\mathsf{ioBPP},\\
\text{one-way functions exist}&\ \Longleftrightarrow\ J_{55\log n}\notin\mathsf{ioBPP}.
\end{aligned}
\]
\end{theorem}

Here $\mathsf{ioBPP}$ requires one randomized polynomial-time procedure to solve every promised input with error at most $1/3$ at infinitely many lengths containing a promised instance. This is worst-case promise hardness, unlike Theorem~\ref{thm:lp}. The gap parameter already organizes two cryptographic conditions in one problem. Extending comparable structure to proof size and circuits is the constructive challenge.

\section{A common measure and an axiom program}\label{sec:program}
A common measure should do more than concatenate these statements. Its domain should specify the object, available information, resource, and success criterion before selecting the desired conclusion. Each resource interpretation should have a theorem identifying the resulting condition. Additional structure should yield invariance under admissible coding, resource monotonicity, or reductions between instances.

\begin{center}\small
\begin{tabularx}{\textwidth}{@{}p{.18\textwidth}X X@{}}
\toprule
Resource & Information object & Target or status\\
\midrule
Proof size & Consistency after a true axiom & BB equivalence with $\ExpAx$; Blindness conjectural\\
Average-case time & $K^t$ values & OWFs: Theorem~\ref{thm:lp}\\
Circuit size & Bit-access descriptions & Polynomial equivalence: Theorem~\ref{thm:kt}\\
Interaction & Interactive descriptions & Key agreement/OWF gaps: Theorem~\ref{thm:blmp}\\
Hierarchy level & Bounded oracle descriptions & Conditional noncollapse: Theorem~\ref{thm:ph}\\
\bottomrule
\end{tabularx}
\end{center}

This is a proposal for a common measure, not a common hypothesis. It is an inventory of interfaces, not a definition of a master constraint. A successful construction would choose the interfaces canonically and prove the correspondences. Distributions, interaction, and nonuniformity must be represented explicitly rather than hidden in the word \emph{access}.

\begin{proposition}\label{prop:reformulation}
Let $\mathcal B$ be a base theory and let $H,Q$ be sentences. If $\mathcal B\vdash\neg Q$, then
\[
\mathcal B\vdash (H\to Q)\leftrightarrow\neg H.
\]
For a uniformly indexed family $H(x),Q(x)$, if $\mathcal B\vdash\forall x\,\neg Q(x)$, then
\[
\mathcal B\vdash
\bigl[\forall x\,(H(x)\to Q(x))\bigr]
\leftrightarrow\forall x\,\neg H(x).
\]
\end{proposition}
\begin{proof}
If $H\to Q$ and $H$ held, then $Q$ would contradict $\neg Q$; hence $\neg H$. Conversely, $\neg H$ implies $H\to Q$. The indexed form is the same propositional argument under the universal quantifier.
\end{proof}

This test is a guardrail for the program. It must be applied without discarding quantifiers or side conditions. For example, if $\mathcal B\vdash\forall p\,(C(p)\to\neg Q(p))$, then the packaged statement $\exists p\,[C(p)\land(H(p)\to Q(p))]$ is equivalent to $\exists p\,[C(p)\land\neg H(p)]$, not in general to the unqualified sentence $\neg H$. A bridge of the form ``efficient success implies an explanation'' contributes no content beyond the desired failure of efficient success whenever the ambient hypotheses already refute the explanation. Such a bridge should be labeled a reformulation, not counted as an independent unifying assumption. A surviving master constraint must evade this collapse through a consequent not already excluded, a quantitative law, or a nonconditional formulation with independently testable consequences.

\paragraph{Mechanical membership.}
Two goals must be distinguished: a decidable syntax of admissible instances with mechanically checked correspondence proofs, and an algorithm deciding the truth of every resulting hardness assertion. The first is the appropriate unification objective. A checker can validate the reduction establishing an instance without deciding a major open conjecture. This is a practical way to pursue mechanical unification without requiring a truth oracle.

\subsection{Extrinsic support and independence}
A potential axiom should earn support through consequences and explanation. We adopt an extrinsic standard of fruitfulness, coherence, and payoff. Tests include a new consequence not built into the definition, robust invariance, an unexpected reduction between established conditions, and an informative account of excluded problems. Intrinsic plausibility motivates these tests but does not replace them. No assertion of self-evidence is needed: one can evaluate a candidate by the precision of the problems it organizes and by deductions that could fail if it were false.

The cryptographic gap parameter provides a useful model of what progress would look like. One does not simply join two desired conclusions: one defines an interactive description problem independently and proves the two correspondences. At the proof-complexity interface, Theorem~\ref{thm:normal} gives a normal form for obstructions, while Theorem~\ref{thm:cert} explains one kind of certified access. A master principle would need a further theorem connecting these interfaces to computational success. Such a theorem must account for random coins, advice, distributions, and the cost of verification.

\paragraph{Testing apparent reliance on hidden advice.}
Can the proposed informational explanation be substantiated by further consequences? A direct route would define an access capability independently of the target hardness claim and prove that a putative efficient solution entails that capability, with the resource costs explicit. The resulting constraint should also yield a consequence not used to formulate it: for example, a reduction between different tasks, a quantitative change when interaction or side information is available, or a new certification barrier. Such deductions would test whether one information principle explains several phenomena rather than merely restating one lower bound. A proved necessity implication would verify a precise sense of reliance; further consequences would supply extrinsic evidence for the proposed constraint, not by themselves prove its truth. This is also a way to distinguish competing information-based explanations.

\begin{proposition}\label{prop:strength}
For fixed polynomial-time presented $\mathcal S{\supseteq}\Sbase$ and $\mathcal U$, let $H_{\mathcal S,\mathcal U}$ be the arithmetic assertion $\mathcal U\not\csim\mathcal S$. With ordinary elementary coding,
\[
\EA\vdash H_{\mathcal S,\mathcal U}\to\Con_{\mathcal S}.
\]
Consequently a consistent effective $\mathcal S{\supseteq}\EA$ with the usual derivability conditions cannot prove $H_{\mathcal S,\mathcal U}$.
\end{proposition}
\begin{proof}
An $\mathcal S$-refutation $w$ gives a proof of any $\psi$ by explosion, with length $|w|+O(|\psi|)$ up to fixed polynomial coding overhead. For $\psi{=}\Con_{\mathcal U}(n)$, binary numerals give target length $O(\log n)$ with a constant depending on $\mathcal U$. Choose a fixed exponent and then an onset depending on $w,\mathcal U$ that absorbs the length into $n^c$. Elementary syntax verifies this construction in $\EA$. Thus inconsistency supplies a polynomial simulation. Contraposition proves the first claim; second incompleteness proves the last.
\end{proof}

This uses an unbounded family of lower bounds, not a single finite assertion. It supplies consistency strength, not independence from ZFC, and does not automatically apply to circuit or average-case instances. Determining any stronger consistency strength is part of the axiom program.

\begin{corollary}\label{cor:independence}
Under the consistency and derivability hypotheses of Proposition~\ref{prop:strength}, $\mathcal S{+}\neg H_{\mathcal S,\mathcal U}$ is consistent. If additionally $\mathcal S$ is sound and $H_{\mathcal S,\mathcal U}$ is true, then both $\mathcal S+H_{\mathcal S,\mathcal U}$ and $\mathcal S{+}\neg H_{\mathcal S,\mathcal U}$ are consistent; equivalently, $H_{\mathcal S,\mathcal U}$ is independent of $\mathcal S$.
\end{corollary}
\begin{proof}
If $\mathcal S{+}\neg H_{\mathcal S,\mathcal U}$ were inconsistent, the deduction theorem and classical logic would give $\mathcal S\vdash H_{\mathcal S,\mathcal U}$, contradicting Proposition~\ref{prop:strength}. Under the additional truth assumptions, $\mathcal S+H_{\mathcal S,\mathcal U}$ is sound and hence consistent.
\end{proof}

\paragraph{Three different independence questions.}
First, a consistent theory $\mathcal S$ cannot prove the fixed assertion $H_{\mathcal S,\mathcal U}$ about its own proof lengths, because this would prove its own consistency. If $\mathcal S$ is sound and that assertion is true, $\mathcal S$ also cannot prove its negation. This yields independence from the constrained theory under the external truth assumption; it does not prove the hardness assertion true.

Second, a fixed sound base $\mathcal B$ can fail to certify an externally existing simulation. A true fixed-bound $\Pi^0_1$ simulation sentence unprovable in $\mathcal B$ is independent of $\mathcal B$ by soundness, but a certification barrier alone does not establish that this sentence is true. Conversely, an external non-simulation result says that all its proposed polynomial certificates are false. These two situations must not be interchanged.

Third, \emph{ambient independence} asks about a fixed natural instance over a separately chosen strong theory $\mathcal A$, for example ZFC. If $\mathcal A$ already proves $\Con_{\mathcal S}$, Proposition~\ref{prop:strength} supplies no obstruction to $\mathcal A$ proving $H_{\mathcal S,\mathcal U}$. Replacing $\mathcal S$ by $\mathcal A$ merely repeats the self-consistency argument at a higher level. The stronger objective is to keep the intended instance fixed and establish both $\mathcal A\nvdash H$ and $\mathcal A\nvdash\neg H$, for example by proving appropriate relative-consistency or model-existence results for $\mathcal A+H$ and $\mathcal A{+}\neg H$. No such ambient independence theorem is claimed here.

\paragraph{Fixing the sentence.}
For fixed $\mathcal S,\mathcal U$, non-simulation has the $\Pi^0_2$ form (1). If $\mathcal U{=}\Sbase{+}\phi_{BB}(k)$, both $k$ and the actual numeral $\BB(k)$ are fixed before this classification. A quantified assertion ranging over true Busy Beaver values requires its own uniform arithmetic encoding; it cannot use $k\mapsto\BB(k)$ as a computable term. An independence program must first specify that encoding, its ordinary proof predicate, and whether its truth claim is standard-model or internal.

\paragraph{What would constitute a payoff?}
A useful axiom candidate could give a new reduction between two currently separate hardness conditions, identify a natural intermediate principle, or explain why an additional resource such as interaction changes the required information gap. A failure of full unification could also be informative if supported by a precise separation rather than by different terminology. Formal independence would strengthen the foundational case, but the program can produce worthwhile mathematics through these structural results before any axiom proposal is settled.

\section{Provable feasibility and its residual}\label{sec:def}
A definitional alternative is to call a task feasible only when its feasibility can be proved in a specified theory, and to regard the remaining tasks as hard relative to that theory. This approach makes the informational limits of formal systems explicit, but it changes the object of study: the residual measures failures of certification rather than necessarily failures of efficient computation. We formalize this distinction and examine when it can nevertheless transfer to limitations on algorithms and circuits. Fix effective task specifications $e$, procedure codes $M$, resource bounds $c$, and an effective syntactic map $C(e,M,c)$ expressing arithmetically that $M$ correctly solves $e$ within the bound $c$. Let $\mathcal T$ be an effectively presented theory.
\begin{definition}
The set $F_{\mathcal T,C}$ consists of tasks $e$ for which $\mathcal T\vdash C(e,M,c)$ for some $M,c$. Its complement $H_{\mathcal T,C}$ is the uncertified residual.
\end{definition}
\begin{proposition}\label{prop:re}
$F_{\mathcal T,C}$ is recursively enumerable and $H_{\mathcal T,C}$ is co-recursively enumerable.
\end{proposition}
\begin{proof}
Dovetail the enumeration of proofs and tuples $e,M,c$, listing $e$ whenever the matching theorem appears. For r.e.\ axioms, include their enumeration witnesses in proof enumeration. Nonmembership need not be decidable.
\end{proof}

Soundness makes certificates semantically correct but does not make all efficient behavior certifiable. The following implementation-level example locates the distinction.

\begin{proposition}\label{prop:zero}
Let $\mathcal T{\supseteq}\EA$ be consistent with a polynomial-time ordinary proof predicate and the usual second incompleteness theorem. The polynomial-time procedure $M_{\mathcal T}(p)=1$ if $p$ is a $\mathcal T$-refutation, and $0$ otherwise, computes the zero function in the standard model. Its universal zero-output correctness is not provable in $\mathcal T$.
\end{proposition}
\begin{proof}
Consistency makes every refutation check fail. By verification of the code, universal zero-output correctness is equivalent over $\mathcal T$ to $\Con_{\mathcal T}$, so second incompleteness applies.
\end{proof}

The zero-function task still has a trivial certified zero algorithm. This proposition concerns one implementation, not every implementation of the task. Task-level hardness requires excluding all certificates in the specified regime, a stronger and representation-sensitive requirement.

\begin{proposition}\label{prop:negativecert}
Let $\mathcal T{\supseteq}\EA$ be consistent and effectively axiomatized with the standard elementary proof transformations and usual derivability conditions. For every fixed arithmetic sentence $\psi$,
\[
\EA\vdash\neg\operatorname{Prov}_{\mathcal T}(\ulcorner\psi\urcorner)
\to\Con_{\mathcal T},
\qquad
\mathcal T\nvdash\neg\operatorname{Prov}_{\mathcal T}(\ulcorner\psi\urcorner).
\]
\end{proposition}
\begin{proof}
Formalized explosion transforms any $\mathcal T$-refutation into a $\mathcal T$-proof of $\psi$. Contraposition gives the implication in $\EA$. If $\mathcal T$ proved its antecedent, it would prove its own consistency, contrary to second incompleteness.
\end{proof}

Taking $\psi$ to be a fixed implementation's correctness and resource assertion gives a precise limitation on identifying the residual internally. It does not prevent an external argument, a conditional argument assuming consistency, or a certificate in a stronger theory. Because it applies to every sentence, this is a generic consistency barrier rather than a task-sensitive computational lower bound.

\begin{proposition}\label{prop:global}
Fix a sound effective base $\mathcal A$. Range over all sound effective extensions, including $\mathcal A{+}\theta$ for every true arithmetic sentence $\theta$. Then
\[
\bigcup_{\mathcal T\supseteq\mathcal A}F_{\mathcal T,C}
=
\{e:\exists M,c\ \N\models C(e,M,c)\},
\]
where the union is over these sound theories.
\end{proposition}
\begin{proof}
Soundness proves the forward inclusion. Conversely, choose a true instance $C(e,M,c)$ and adjoin that sentence to $\mathcal A$. The resulting sound effective theory certifies $e$.
\end{proof}

If $C$ exactly expresses semantic feasibility, the global residual is ordinary infeasibility. This makes precise the metaphor of hardness as the shadow cast by mathematics as a whole. It also locates the research challenge: a new undefinability theorem must use natural structure in the specifications, theories, or certificates. Semantic quantification over sound theories alone does not prove such a theorem, because the displayed equality can eliminate that quantification.

Do algorithms and circuits inherit the blind spots of the theory that proves their validity? Here validity includes correctness and the claimed resource bound. The question is whether a theory-certified efficient solution would also certify information inaccessible to that theory. The following internal-implication criterion makes this precise without identifying external existence with provability.

\begin{proposition}\label{prop:spillover}
Fix a theory $\mathcal B$, arithmetic sentences $C_a$ describing the correctness and resource bound of named implementations, and fixed simulation sentences $D_{x(a)}$. Suppose $\mathcal B\vdash C_a\to D_{x(a)}$ for each index $a$ under consideration. If $\mathcal B\nvdash D_{x(a)}$, then $\mathcal B\nvdash C_a$. In particular, with $D_x{=}\Sim^{c_x,n_x}_{\mathcal S,\mathcal S+(x\in\R)}$ and the hypotheses of Theorem~\ref{thm:certkh}, the conclusion holds whenever $x(a)$ is a true random string outside that theorem's finite exceptional set.
\end{proposition}
\begin{proof}
A $\mathcal B$-proof of $C_a$, followed by the hypothesized implication, would be a $\mathcal B$-proof of $D_{x(a)}$. The specialization uses the cited finiteness theorem, whose exceptional set is uniform over the fixed-bound parameters.
\end{proof}

The logical transfer is elementary; the substantive research question is which natural algorithms or circuit constructions supply its premise. For an algorithm, $C_a$ must specify its task, correctness, and time bound, rather than just termination. For a randomized algorithm it must include the distribution and success probability. For circuits it must say whether the object is one finite circuit, an effective circuit generator, or a definable family: an arbitrary nonuniform family need not have an effective arithmetic name. A reduction must also preserve the information parameter and the proof-size scale needed by the application.

\paragraph{Oracle-matched certification.}
An ordinary mathematical proof may legitimately use principles unavailable to the algorithm or circuit whose correctness it establishes: the implementation does not execute its proof. The issue changes when the proof is offered as an explanation of the implementation's informational resources. A certifying theory with an oracle unavailable to the implementation could conceal precisely the information cost the framework is intended to expose. This suggests an oracle-matching discipline: either require an oracle-free certificate for an oracle-free implementation, or record the oracle $O$ explicitly and state a relativized result in which the implementation, its certifier, and the comparison class receive the declared resources. Whether this discipline is necessary, and which conclusions survive it, are research questions rather than changes to the standard definition of algorithmic correctness.

The Busy Beaver instance gives a focused test. With a fixed effective machine coding, the sequence of all values $\BB(k)$ is Turing-equivalent to the halting problem: the values decide halting after padding a machine to a known state bound and running it for the corresponding bound, while a halting oracle identifies the finitely many halting $k$-state machines and hence their maximum running time. A decision predicate for all true $\Pi^0_1$ sentences has the same Turing degree, since nonhalting assertions are $\Pi^0_1$ and counterexamples to $\Pi^0_1$ sentences are computably enumerable. One may therefore ask whether either the existence or the nonexistence of an optimal proof system is independent of every theory in a specified natural class of sound effective oracle-free theories, but becomes decidable in an appropriately oracle-augmented theory. No such result is claimed here. The phrase ``requires a Busy Beaver oracle'' must be understood through such a fixed class and a relative-consistency or model-theoretic statement; a single finite proof does not itself reveal an oracle requirement.

This formulation also controls a possible circularity. A theory whose axiom predicate uses the halting oracle is not an ordinary Cook--Reckhow proof system, whose verifier must run in deterministic polynomial time~\cite[Definition~1.3]{CookReckhow}. If oracle proof verifiers are admitted, optimality must be relativized and all competitors must be compared relative to the same oracle. Thus a proof in an oracle theory cannot count as evidence for an unrelativized optimal system merely because the certifying theory has been allowed the missing information. The object theory, certifying theory, oracle, and comparison class must remain separate parameters.

This distinction clarifies the definitional proposal. Under the certified notion, one can prove limitations of a specified mathematical verification regime even when external feasibility is unsettled. To recover a classical lower bound, one additionally needs a task-specific theorem that every efficient solution admits a certificate of the required form. That completeness property is an objective, not part of the definition of ordinary polynomial time.

\section{Conclusion}\label{sec:conclusion}
This paper establishes a proof-theoretic core for the information-constraint program while keeping its broader claims explicitly conjectural. The nonexistence of optimal proof systems admits an unconditional consistency-extension characterization, Busy Beaver values provide a canonical family of consistency information, and certified simulation unconditionally entails relative consistency. The central question for the program is whether limited access to mathematical information supplies a genuine and transferable explanation of hardness. Its immediate tests are which restricted external simulations admit weak-base explanations; whether a canonical measure can accommodate the different resource and quantifier structures in Section~\ref{sec:meta}; and whether that measure yields new consequences not already built into its definitions.

Can apparent reliance on hidden advice be made precise and tested through additional consequences? In particular, can the same constraint explain a second hardness phenomenon or predict a resource tradeoff without incorporating that conclusion as another premise? Such results would give mathematical substance to the possibility that missing information constraints account for some persistent open problems.

Further objectives are natural task-level unprovability theorems and robust independence of core instances from $ZFC$ or stronger theories. A separate robustness question concerns randomness. Allender, Friedman and Gasarch~\cite{AllenderFG2013} prove upper bounds for decidable languages reducible to prefix-complexity randomness sets, after intersecting over the choice of universal prefix machine. Those theorems do not directly apply to the single-machine plain-complexity set $\R$ used here. Stability of its certification and finite-scale claims under machine changes requires specified translations and threshold adjustments. A successful unification should also explain which computational models and problems fall outside it and what additional information principle would be needed to include them. The established metacomplexity correspondences make this an attractive direction. The certification theorem supplies a precise arithmetic instance of the demand that efficient use of information have an explanation.

\bibliographystyle{amsplain}
\bibliography{equivalence}

\clearpage
\appendix
\section{Pairwise information and no mutual help}\label{app:pairwise}
Kolmogorov Blindness at a fixed eligible theory gives exponential lower bounds even after adjoining another true randomness axiom. Universal Blindness also implies factual nonaccess through an auxiliary theory. These consequences use the relation $x\perp_{\log}y$ defined in Section~\ref{sec:kh}; conditional complexity alone does not give factual nonaccess.

\begin{proposition}[Factual nonaccess under universal Blindness]\label{prop:blindness-nonaccess}
Assume Kolmogorov Blindness for every sound finitely axiomatized sequential theory extending $\Sbase$. For every fixed sound $\mathcal S{\supseteq}\Sbase$ with a polynomial-time ordinary proof predicate and the specified feasible presentation, there is $k_{\mathrm{access}}^{\mathcal S}$ such that, for all $x\perp_{\log}y$ with $|x|{=}|y|{>}k_{\mathrm{access}}^{\mathcal S}$, $\EA{+}\Con_{\mathcal S+(y\in\R)}\nvdash x\in\R$, and the symmetric nonprovability holds. Consequently, $\EA$ proves neither $\Con_{\mathcal S+(y\in\R)}{\to}\Con_{\mathcal S+(x\in\R)}$ nor its converse.
\end{proposition}
\begin{proof}
Following the companion paper~\cite[Proposition~5.6]{MonroeCharacterizing}, fix a sound finitely axiomatized sequential base $A$ containing $\EA$ and $\Sbase$ in the common definitional language. Put $\chi_{\mathcal S}{:=}\forall z\,(z\in\R{\to}\Con_{\mathcal S+(z\in\R)})$ and $\mathcal T{:=}A{+}\chi_{\mathcal S}$. Soundness of $\mathcal S$ makes $\chi_{\mathcal S}$ true, so $\mathcal T$ is sound, finitely axiomatized and sequential. Apply Blindness to $\mathcal T$ and set $k_{\mathrm{access}}^{\mathcal S}{:=}k_R^{\mathcal T}$.

Suppose a mutually random pair above this cutoff satisfies $\EA{+}\Con_{\mathcal S+(y\in\R)}\vdash x\in\R$. The finite sound theory $Q{:=}\mathcal T{+}(y\in\R)$ proves $\Con_{\mathcal S+(y\in\R)}$ using $\chi_{\mathcal S}$, and hence proves $x\in\R$. Replacing the axiom $x\in\R$ by this fixed $Q$-proof gives a base-verified polynomial refutation translation from $\mathcal T{+}(x\in\R)$ to $Q$. The finitely axiomatized sequential theory $Q$ has polynomial-size proofs of its own finite consistency (Pudl\'ak~\cite[Theorem~3.1(2)]{Pudlak}); Proposition~\ref{prop:positive} therefore gives polynomial-size $Q$-proofs of $\Con_{\mathcal T+(x\in\R)}(n)$. This contradicts Blindness at $\mathcal T$ with helper $y$, since $K_U(x\mid y){\ge}\ell(|x|)$ and $2^{\eta_{\mathcal T}n}$ eventually exceeds every polynomial. Interchange $x,y$ for the symmetric assertion. The relative-consistency obstructions follow from Lemma~\ref{lem:random}.
\end{proof}

\begin{lemma}\label{lem:pairdensity}
For sufficiently large $m$, at least $(1-O(m^{-d}))2^{2m}$ ordered pairs of length-$m$ strings satisfy $x\perp_{\log}y$.
\end{lemma}
\begin{proof}
For each fixed $y$, fewer than $2^{\ell(m)}$ strings $x$ have conditional complexity below $\ell(m)$. Summing over $y$ bounds the failing pairs by $O(2^{2m}/m^d)$. The same estimate holds with $x,y$ interchanged. The two failures of unconditional randomness have the same bounds by ordinary program counting. A union bound over these four events proves the claim.
\end{proof}

\begin{theorem}[Pairwise density under Blindness]\label{thm:pairwise}
Assume Kolmogorov Blindness at fixed sound finitely axiomatized sequential $\mathcal S{\supseteq}\Sbase$ with its stipulated presentation. For every sufficiently large fixed $m{>}k_R^{\mathcal S}$, at least a $1-O(m^{-d})$ fraction of ordered pairs of length $m$ are true mutually random pairs and satisfy both of the following for every $n{>}N_{R,m}^{\mathcal S}$: $\Con_{\mathcal S+(x\in\R)}(n)$ has no $\mathcal S+(y\in\R)$-proof of size at most $2^{\eta_{\mathcal S}n}$, and the same bound holds with $x,y$ interchanged.
\end{theorem}
\begin{proof}
Apply Conjecture~\ref{conj:kolmogorov-blindness} to $(x,y)$ and $(y,x)$. Each short proof would contradict the corresponding conditional randomness bound. Both applications have the same onset. Lemma~\ref{lem:pairdensity} supplies the fraction of eligible pairs.
\end{proof}

The target remains consistency of the original $x$-extension, with $y\in\R$ as helper. Density is over ordered pairs. Bounds on the common onset and bridges to the computational conditions in Section~\ref{sec:meta} remain open.

\end{document}